\documentclass[a4paper,12pt]{article}
\usepackage[latin1]{inputenc}
\usepackage[T1]{fontenc}
\usepackage{amsmath, amsthm, amssymb}
\usepackage{array}

\usepackage[final]{graphicx}   
\graphicspath{{tikz/}}

\newtheorem{definition}{Definition}[section]
\newtheorem{theorem}[definition]{Theorem}
\newtheorem{lemma}[definition]{Lemma}
\newtheorem{remark}[definition]{Remark}

\newtheorem{corollary}[definition]{Corollary}

\newtheorem{example}[definition]{Example}
\newtheorem{conjecture}[definition]{Conjecture}

\begin{document}

\pagenumbering{roman}
\pagestyle{empty} 
\title{Necessary Field Size and Probability for MDP and Complete MDP Convolutional Codes}
\author{Julia Lieb}
\maketitle



  
\setcounter{page}{1}


\pagenumbering{arabic}
\pagestyle{plain} 

\begin{abstract}
It has been shown that maximum distance profile (MDP) convolutional codes have optimal recovery rate for windows of a certain length, when transmitting over an erasure channel. In addition, the subclass of complete MDP convolutional codes has the ability to reduce the waiting time during decoding.
In this paper, we derive upper bounds on the necessary field size for the existence of MDP and complete MDP convolutional codes and show that these bounds improve the already existing ones. Moreover, we derive lower bounds for the probability that a random code is MDP respective complete MDP.
\end{abstract}

\section{Introduction}	
Convolutional codes play an important role for digital communication. When considering the erasure channel, which is the most used channel in multimedia traffic, these codes can correct more errors than the classical block codes.

Besides the classical free distance, convolutional codes possess a different notion of distance, called column distance. The column distances of a convolutional code are limited by an upper bound, which was proven in \cite{RS99}. Convolutional codes attaining these bounds, i.e. convolutional codes whose column distances increase as rapidly as possible for as long as possible are called maximum distance profile (MDP) codes. These codes were introduced in \cite{strongly} and are especially suitable for the
use in sequential decoding algorithms. In \cite{vp}, the authors showed that MDP convolutional codes can correct the maximum possible number of errors in some sliding window of a certain length (depending on the code parameters). Moreover, they considered reverse MDP convolutional codes, which have the advantage that optimal error correction is possible with forward and backward decoding algorithms. Finally, complete MDP convolutional codes, which are again a subclass of reverse MDP convolutional codes, have the additional benefit that they can correct even more error patterns than reverse MDP convolution codes, e.g. there is less waiting time when a large burst of erasures occurs and no correction is possible for some time \cite{vp}.

The existence (and genericity) of reverse MDP convolutional codes for all code parameters has been proven in \cite{vp}. In \cite{cmdp}, it has been shown that for the existence of an $(n,k,\delta)$ complete MDP convolutional code, it is necessary to have $(n-k)\mid\delta$ and that complete MDP convolutional codes exist (and are generic) for all code parameters fulfilling this condition. The case $(n-k)\nmid\delta$, in which a complete MDP convolutional code cannot exist, is also much more involved when considering just MDP convolutional codes, see \cite{dr}.
There are some general constructions for MDP \cite{strongly}, \cite{dr13} and complete MDP \cite{cmdp} convolutional codes. However, all of these constructions have the disadvantage that they only work over base fields of very large size.

This provokes the question for the minimal field size such that an MDP respective complete MDP convolutional code could exist. For the case of MDP convolutional codes, there is something done to solve this problem in \cite{b}, where the authors provide an upper bound on the necessary field size. In \cite{b} as well as in \cite{strongly}, where an - until now unproven - conjecture about a bound on the necessary field size is raised, superregular Toeplitz matrices are used, i.e. the question is connected to the problem of determining the necessary field size for the existence of such superregular Toeplitz matrices. In this paper, we improve these bounds by other means than using superregular Toeplitz matrices. For complete MDP convolutional codes, the so far only result on the necessary field size could be derived from the constructions in \cite{cmdp} but this is leading to very weak bounds. In this paper, we also present bounds for the necessary field size for complete MDP convolutional codes. 

Since constructions - especially over fields of possibly small size - have been found to be very hard to obtain, it is an interesting question, how large the probability for an MDP respective complete MDP convolutional code is, when choosing the code randomly. In this paper, we give lower bounds for this probability for MDP as well as for complete MDP convolutional codes.

The paper is structured as follows. In Section 2, we start with some preliminaries about MDP convolutional codes. In Section 3, we give the exact minimum field size for $(n,1,1)$ (and $(n,n-1,1)$) MDP, reverse MDP and complete MDP convolutional codes as well as the corresponding probabilities.
In Section 4, we show upper bounds for the necessary field size for MDP convolutional codes and lower bounds for the probability that a convolutional code is MDP. 
In Section 5, we generalize the results of Section 4 to complete MDP convolutional codes.
Section 6 provides an improved bound on the field size for MDP convolutional codes in the case $\delta<\max\{k,n-k\}$. In Section 7, we show that except for very few choices of (small) parameters, the new bounds for the field size of this paper are better than all bounds existing up to now.

\section{MDP Convolutional Codes}

In this section, we summarize the basic definitions and properties concerning MDP convolutional codes.
One way to define a convolutional code is via polynomial generator matrices.

\begin{definition}\ \\
A \textbf{convolutional code} $\mathfrak{C}$ of \textbf{rate} $k/n$ is a free $\mathbb F[z]$-submodule of $\mathbb F[z]^n$ of rank $k$.
We refer to it as $(n,k,\delta)$ convolutional code.\\
There exists $G(z)\in\mathbb F[z]^{n\times k}$ of full column rank such that
$$\mathfrak{C}=\{v(z)\in\mathbb F[z]^n\ |\ v(z)=G(z)m(z)\ \text{for some}\ m(z)\in\mathbb F[z]^k\}.$$
$G(z)$ is called \textbf{generator matrix} of the code and is unique up to right multiplication with a unimodular matrix $U(z)\in Gl_k(\mathbb F[z])$.\\
The \textbf{degree} $\delta$ of $\mathfrak{C}$ is defined as the maximal degree of the $k\times k$-minors of $G(z)$.
Let $\delta_1,\hdots, \delta_k$ be the column degrees of $G(z)$. Then, $\delta\leq\delta_1+\cdots+\delta_k$ and if $\delta=\delta_1+\cdots+\delta_k$, $G(z)$ is called a \textbf{minimal} generator matrix.
\end{definition}



There is a generic subclass of convolutional codes that could not only be described by an image representation via generator matrices but also by a kernel representation via the so-called parity-check matrices, which will be introduced in the following. Therefore, we need the notion of right prime and left prime polynomial matrices.

\begin{definition}\ \\
Let $\overline{\mathbb F}$ denote the algebraic closure of $\mathbb F$.
A polynomial matrix $G(z)\in\mathbb F[z]^{n\times k}$ with $k<n$ is called \textbf{right prime} if it has full column rank for all $z\in\overline{\mathbb F}$. For $k>n$, it is called \textbf{left prime} if it has full row rank for all $z\in\overline{\mathbb F}$. 
\end{definition}


\begin{definition}\ \\
A convolutional code $\mathfrak{C}$ is called \textbf{non-catastrophic} if one and therefore, each of its generator matrices is right prime.
\end{definition}

 

\begin{definition}\ \\
If $\mathfrak{C}$ is non-catastrophic, there exists a so-called \textbf{parity-check matrix} $H(z) \in\mathbb F[z]^{(n-k)\times n}$ of full rank, such that
$$\mathfrak{C} =\{ v(z)\in \mathbb F[z]^n\ |\ H(z)v(z) = 0 \in \mathbb F[z]^{n-k}\}.$$
Clearly, a parity-check matrix of $\mathfrak{C}$ is not unique and it is possible to choose it left prime and row proper. In this case, the sum of the row degrees of $H(z)$ is equal to the degree $\delta$ of $\mathfrak{C}$ \cite{con}.\\
$H(z)$ has \textbf{generic row degrees} if $\nu=\lceil\frac{\delta}{n-k}\rceil$ and the first $\delta-(n-k)(\nu-1)$ row degrees of $H(z)$ are equal to $\nu$ and the remaining $(n-k)\nu-\delta$ row degrees are equal to $\nu-1$.
\end{definition}

\begin{remark}\label{kh}\ \\
Allowing permutation of the entries of the codeword $v(z)$ respective of the columns of the parity-check matrix $H(z)$, each non-catastrophic convolutional code has a unique parity-check-matrix of the form  $H(z)=[P(z)\ Q(z)]$, where $P$ and $Q$ are left coprime, $Q$ is of Kronecker-Hermite form, i.e.
$q_{ii}$ monic, $deg(q_{ji})<\deg(q_{ii})$ for $j\neq i$,\\
$deg(q_{ij})<\deg(q_{ii})$ for $j<i$ and  $deg(q_{ij})\leq \deg(q_{ii})$ for $j>i$,

 and the row degrees of $P$ are at most equal to the row degrees of $Q$.
\end{remark}

We will need the representation by parity-check matrices to define complete MDP convolutional codes.
Bur first of all, we want to introduce MDP convolutional codes, for which we have to consider distances of convolutional codes.

\begin{definition}\ \\
The \textbf{Hamming weight} $wt(v)$ of $v\in\mathbb F^n$ is defined as the number of its nonzero components.\\
For $v(z)\in\mathbb F[z]^n$ with $\deg(v(z))=\gamma$, write $v(z)=v_0+\cdots+v_{\gamma}z^{\gamma}$ with $v_t\in\mathbb F^n$ for $t=0,\hdots,\gamma$ and set $v_t=0\in\mathbb F^n$ for $t\geq\gamma+1$. Then, for $j\in\mathbb N_0$, the \textbf{j-th column distance} of a convolutional code $\mathfrak{C}$ is defined as
$$d_j^c(\mathfrak{C}):=\min_{v(z)\in\mathfrak{C}}\left\{\sum_{t=0}^j wt(v_t)\ |\ v(z)\not\equiv 0\right\}.$$
\end{definition}


There exist upper bounds for the column distances of a convolutional code.

\begin{theorem}\cite{strongly}\label{ub}\ \\
 $$d_j^c (\mathfrak{C}) \leq (n-k)(j + 1) + 1\qquad\text{for}\quad j\in\mathbb N_0$$
\end{theorem}

We are interested in convolutional codes with good distance properties, i.e. in those codes that reach the bounds of the preceding theorem.

\begin{definition}\cite{mdp}\ \\
A convolutional code $\mathfrak{C}$ of rate $k/n$ and degree $\delta$ has
  \textbf{maximum distance profile (MDP)} if 
$$d_j^c(\mathfrak{C})=(n-k)(j+1)+1\quad \text{for}\ j=0,\hdots,L:=\left\lfloor\frac{\delta}{k}\right\rfloor+\left\lfloor\frac{\delta}{n-k}\right\rfloor$$
\end{definition}

As mentioned in the introduction, MDP convolutional codes have the property that their column distances increase as rapidly as possible for as long as possible. Indeed, $j=L$ is the largest possible value for which $d_j^c$ can attain the upper bound from Theorem \ref{ub}. Moreover, according to \cite{strongly}, it is sufficient to have equality for $j=L$ in Theorem  \ref{ub} to get an MDP convolutional code.


In the following, we will provide criteria to check whether a convolutional code has a maximum distance profile. Therefore, we need the notion of trivially zero determinants.


\begin{definition}\ \\
(i) Let $n\in\mathbb N$ and $A\in\mathbb F^{n\times n}$ be a matrix with the property that each of its entries is either fixed to zero or is a free variable from $\mathbb F$. Its determinant $\det(A)$ is called \textbf{trivially zero} if it is zero for all choices for the free variables in $A$.\\
(ii) An $n\times n$ Toeplitz matrix of the form $\left(\begin{array}{ccc}
a_1 &  & 0 \\ 
\vdots & \ddots &  \\ 
a_n & \hdots & a_1
\end{array} \right)$ is called \textbf{superregular} if all its minors that are not trivially zero are nonzero.
\end{definition}

\begin{theorem}\cite{strongly}\ \\
Let the convolutional code $\mathfrak{C}$ be generated by a right prime minimal polynomial matrix $G(z)=\sum_{i=0}^{\mu}G_iz^i\in\mathbb F[z]^{n\times k}$ and have the left prime and row proper parity-check matrix $H(z)=\sum_{i=0}^{\nu}H_iz^i\in\mathbb F[z]^{(n-k)\times n}$. Equivalent are:
\begin{itemize}
\item[(a)] $\mathfrak{C}$ is of maximum distance profile.
\item[(b)] $\mathcal{G}_L:=\left[\begin{array}{ccc} G_0 & & 0\\ \vdots & \ddots &  \\ G_L & \hdots & G_0 \end{array}\right]$ where $G_i=0$ for $i>\mu$ has the property that every full size minor that is not trivially zero, i.e. zero for all choices of $G_1,\hdots,G_L$, is nonzero.
\item[(c)] $\mathcal{H}_L:=\left[\begin{array}{ccc} H_0 & & 0\\ \vdots & \ddots &  \\ H_L & \hdots & H_0 \end{array}\right]$ where $H_i= 0$ for $i>\nu$ has the property that every full size minor that is not trivially zero is nonzero.
\end{itemize}
\end{theorem}

\begin{remark}\ \\
The not trivially zero full size minors of $\mathcal{H}_L$ are exactly those which are formed by columns with indices $1\leq j_1<\cdots<j_{(L+1)(n-k)}\leq(L+1)n$ which fulfil $j_{s(n-k)}\leq sn$ for $s=1,\hdots,L$.
\end{remark}

The following duality result for MDP convolutional codes will be important at many points of this paper.

\begin{theorem}\cite{strongly}\label{dual}\ \\
An $(n,k,\delta)$ convolutional code is MDP if and only if its dual code, which is an $(n,n-k,\delta)$ convolutional code, is MDP.
\end{theorem}

Next, we introduce reverse MDP convolutional codes, which are advantageous for use in forward and backward decoding algorithms \cite{vp}.

\begin{definition}\cite{h}\ \\
Let $\mathfrak{C}$ be an $(n,k,\delta)$ convolutional code with right prime minimal generator matrix $G(z)$, which has entries $g_{ij}(z)$. Set $\overline{g_{ij}(z)}:=z^{\delta_j}g_{ij}(z^{-1})$. Then, the code $\overline{\mathfrak{C}}$ with generator matrix $\overline{G(z)}$, which has $\overline{g_{ij}(z)}$ as entries, is also an $(n,k,\delta)$ convolutional code, which is called the \textbf{reverse code} to $\mathfrak{C}$.\\
It holds: $v_0+\cdots+v_dz^d\in\overline{\mathfrak{C}}\ \Leftrightarrow\ v_d+\cdots+v_0z^d\in\mathfrak{C}$.
\end{definition}

\begin{definition}\cite{vp}\ \\
Let $\mathfrak{C}$ be an MDP convolutional code. If $\overline{\mathfrak{C}}$ is also MDP, $\mathfrak{C}$ is called \textbf{reverse MDP} convolutional code.
\end{definition}

\begin{remark}\cite{vp}\ \\
Let $(n-k)\mid\delta$ and $H(z) = H_0 + \cdots +H_{\nu}z^{\nu}$ be a left prime and row proper parity-check matrix of the MDP code $\mathfrak{C}$. Then the reverse
code $\overline{\mathfrak{C}}$ has parity-check matrix $\overline{H(z)} = H_{\nu} +\cdots +H_0z^{\nu}$. Therefore, $\mathfrak{C}$ is reverse MDP if and only if every full size minor of the matrix
$$\mathfrak{H}_L:=\left[\begin{array}{ccc} H_{\nu} & \cdots & H_{\nu-L}\\  & \ddots & \vdots \\ 0 &  & H_{\nu} \end{array}\right]$$
formed from the columns with indices $j_1,\hdots,j_{(L+1)(n-k)}$
with $j_{s(n-k)+1} > sn$, for $s = 1,\hdots,L$ is nonzero.
\end{remark}

Next, we introduce complete MDP convolutional codes, which are even more advantageous for decoding than reverse MDP convolutional codes \cite{vp}.

\begin{definition}\cite{vp}\label{com}\ \\
Let $H(z)=H_0+H_1z+\cdots H_{\nu}z^{\nu}\in\mathbb F[z]^{(n-k)\times n}$ be a parity-check matrix of the convolutional code $\mathfrak{C}$ of rate $k/n$. Set $L:=\lfloor\frac{\delta}{n-k}\rfloor+\lfloor\frac{\delta}{k}\rfloor$. Then
\begin{align}\label{ppc}
\mathfrak{H}:=\left(\begin{array}{ccccc}
H_{\nu} & \cdots & H_0 &   & 0 \\ 
  & \ddots &   & \ddots &   \\ 
0 &   & H_{\nu} & \cdots & H_0
\end{array}\right)  \in\mathbb F^{(L+1)(n-k)\times (\nu+L+1)n}
\end{align}
is called \textbf{partial parity-check matrix} of the code. Moreover, $\mathfrak{C}$ is called \textbf{complete MDP} convolutional code if for any of its parity-check matrices $H(z)$, every full size minor of $\mathfrak{H}$ which is not trivially zero is nonzero.
\end{definition}

\begin{remark}\ \\
(i) Every complete MDP convolutional code is a reverse MDP convolutional code. \cite{vp}\\
(ii) A complete MDP convolutional code exists over a sufficiently large base field if and only if $(n-k)\mid\delta$. \cite{cmdp}
\end{remark}

As for $\mathcal{H}_L$ - when considering MDP convolutional codes - and additionally for $\mathfrak{H}_L$ - when considering reverse MDP convolutional codes - one could describe the not trivially zero full size minors of the partial parity-check matrix $\mathfrak{H}$ by conditions on the indices of the columns one uses to form the corresponding minor.


\begin{lemma}\cite{vp}\label{index}\ \\
A full size minor of $\mathfrak{H}$ formed by the columns $j_1,\hdots,j_{(L+1)(n-k)}$ is not trivially zero if and only if 
\begin{itemize}
\item[(i)]
$j_{(n-k)s+1}>sn$
\item[(ii)]
$j_{(n-k)s}\leq sn+\nu n$
\end{itemize}
for $s=1,\hdots,L$.

This is equivalent to $j_1\in\{1,\hdots,\nu n+k+1\},\hdots,j_{n-k}\in\{n-k,\hdots,(\nu+1)n\}$, $j_{n-k+1}\in\{n+1,\hdots,(\nu+1)n+k+1\},\hdots,j_{(n-k)(L+1)}\in\{(L+1)n-k,\hdots,(\nu+1+L)n\}$.
\end{lemma}
%
%

Throughout this paper, we will use the following notations: For a finite field $\mathbb F$, we set $t:=|\mathbb F|^{-1}$. Moreover, we denote a real valued function $f(t)$ in the variable $t$ by $O(t^n)$ for some $n\in\mathbb N$ if $\lim_{t\rightarrow 0}\frac{f(t)}{t^n}\leq C$ for some constant $C\in\mathbb R$.
Moreover, the following theorem will be used frequently throughout this paper:

\begin{theorem}[Schwartz-Zippel]\cite[Corollary 1]{sch}\label{sz}\ \\
(a) For $r\in\mathbb N$, consider $f\in\mathbb F[x_1,\hdots,x_r]$ with total degree $d\geq 0$. Then, $f$ has at most $d\cdot|\mathbb F|^{n-1}$ zeros.\\
(b) Let $f\in\mathbb F[x_1,\hdots,x_r]$ be a nonzero polynomial of total degree $d$. Moreover, let $v_1,\hdots, v_r$ be selected at random independently and uniformly from $\mathbb F$. Then, the probability that $(v_1,\hdots, v_r)$ is a zero of $f$ is at most $d\cdot t$.
\end{theorem}

\section{Results for $(n,n-1,1)$ convolutional codes}

As a starting point, in this section, we want to consider unit memory convolutional codes of high rate, i.e. $\delta=1$ and $k=n-1$. According to Theorem \ref{dual}, these codes are dual to the $(n,1,1)$ convolutional codes, which should therefore also be treated in this section. With this choice of parameters one has $L=\lfloor\frac{\delta}{k}\rfloor+\lfloor\frac{\delta}{n-k}\rfloor=1+\lfloor\frac{1}{n-1}\rfloor$. Hence, $L=1$ for $n\geq 3$ and $L=2$ for $n=2$.

\subsection{$(2,1,1)$ convolutional codes}

\begin{theorem}\cite[Theorem 84]{diss}\ \\
A $(2,1,1)$ convolutional code is  MDP if and only if it holds for its generator matrix $G(z):=\sum_{i=0}^{\delta}g_iz^i=g_0+g_1z$ with $g_0,g_1\in\mathbb F^2$ that 
$0\notin\{g_{0,1}, g_{0,2},g_{1,1},g_{1,2}\}$ and $g_{1,1}g_{0,2}-g_{1,2}g_{0,1}\neq 0$.\\ 
Consequently, the probability that a random polynomial matrix $G(z)\in\mathbb F[z]^2$ with $\deg(G(z))=1$ generates a $(2,1,1)$ MDP convolutional code is $\frac{(1-t)^2(1-2t)}{1+t}$.
\end{theorem}

Form this theorem it follows that the number of $G(z)\in\mathbb F[z]^2$ with $\deg(G(z))=1$ that generate a $(2,1,1)$ MDP convolutional code over $\mathbb F$ is $(|\mathbb F|-1)^3\cdot (|\mathbb F|-2)$. Since two such generator matrices generate the same code if and only if they differ by a factor from $\mathbb F\setminus\{0\}$, the number of $(2,1,1)$ MDP convolutional codes over $\mathbb F$ is $(|\mathbb F|-1)^2\cdot (|\mathbb F|-2)$. In particular, there exists a $(2,1,1)$ MDP convolutional code over $\mathbb F$ if and only if $|\mathbb F|\geq 3$.

Next, we want to investigate reverse and complete MDP convolutional codes with these parameters.

\begin{remark}\ \\
The dual of a $(2,1,1)$ convolutional code is again a $(2,1,1)$ convolutional code and it is easy to see that one could formulate the criterion for the MDP property in the same way if using the parity-check matrix: If the code has parity-check matrix $H(z)=\sum_{i=0}^{\delta}h_iz^i=h_0+h_1z$ with $h_0,h_1\in\mathbb F^{1\times 2}$, the code is MDP if and only if $0\notin\{h_{0,1}, h_{0,2},h_{1,1},h_{1,2}\}$ and $h_{1,1}h_{0,2}-h_{1,2}h_{0,1}\neq 0$. 
\end{remark}

\begin{corollary}\ \\
A $(2,1,1)$ convolutional code is MDP if and only if it is complete MDP. Thus, the statements of the preceding theorem are also true for $(2,1,1)$ reverse convolutional codes and $(2,1,1)$ complete MDP convolutional codes.
\end{corollary}

\begin{proof}\ \\
It is easy to the that the conditions on the parity-check matrix of the preceding remark are also sufficient to get a complete (and hence also a reverse) MDP convolutional code.
\end{proof}

At the end of this subsection, we want to compute the probability of a $(2,1,1)$ MDP convolutional code under the condition that the code is non-catastrophic.

\begin{corollary}\ \\
The probability that a non-catastrophic $(2,1,1)$ convolutional code is MDP, reverse MDP or complete MDP is $\frac{(1-t)(1-2t)}{1+t}$.
\end{corollary}

\begin{proof}\ \\
The conditions on the generator matrix $G(z)$ to get an MDP, reverse MDP or complete MDP convolutional code imply that the two entries of $G(z)$ are of degree one and have a different zero. This means that the entries of $G(z)$ are coprime. Thus, each $(2,1,1)$ MDP convolutional code is non-catastrophic. Consequently, to obtain the conditional probability (under the condition that the code is non-catastrophic), one has just to divide the probability of the first theorem by the probability of non-catastrophicity, which is $1-t$; see \cite{ich}.
\end{proof}

\subsection{$(n,1,1)$ convolutional codes for $n\geq 3$}

For codes with these parameters, we consider generator matrices of the form $G(z)=\sum_{i=0}^{\delta}g_iz^i=g_0+g_1z$ with $g_0,g_1\in\mathbb F^n$.

\begin{theorem}\cite{diss}\ \\
For $n\geq 3$, the probability that $G(z)\in\mathbb F[z]^n$ with $\deg(G(z))=1$ generates an $(n,1,1)$ MDP convolutional code is 
$$(1-t^n)^{-1}(1-t)^{n+1}\prod_{i=2}^{n-1}(1-it).$$
\end{theorem}

\begin{remark}\ \\
According to the proof of the preceding theorem, for $n\geq 3$, the number of $(n,1,1)$ MDP convolutional codes over $\mathbb F$ is $|\mathbb F|(|\mathbb F|-1)^{n+1}\prod_{i=2}^{n-1}(|\mathbb F|-i)$, i.e. such a code exists if and only if $|\mathbb F|\geq n$, and one can construct a generator matrix of such a code as follows: choose all entries of $g_0$ arbitrary but nonzero, choose the first entry of $g_1$ arbitrary and then choose entry $i$ of $g_1$ such that $\begin{pmatrix} g_{0,i}\\ g_{1,i}\end{pmatrix}$ is linear independent to $\begin{pmatrix} g_{0,j}\\ g_{1,j}\end{pmatrix}$ for $j\in\{1,\hdots,i-1\}$.\\
One could see that all MDP codes with these parameters are reverse MDP: For $j=1,\hdots,n$, it holds $\overline{g_j}(z)=\begin{cases}g_j(z), & g_{1,j}=0\\
g_{1,j}+g_{0,j}z, & g_{1,j}\neq 0 \end{cases}$. Therefore, $\overline{g_{0,j}}\neq 0$ and it is easy to see that the other conditions are also fulfilled for $\overline{G}$.\\
However, since $n-1\geq 2>\delta$ and thus $(n-k)\nmid\delta$ a complete MDP convolutional code with these parameters cannot exist \cite{cmdp}.
\end{remark}

As in the preceding subsection, we finally consider the probability for MDP convolutional codes under the condition that the code is non-catastrophic.

\begin{theorem}\ \\
The probability that a non-catastrophic $(2,1,1)$ convolutional code is MDP or reverse MDP is $((1-t^n)(1-t^{n-1}))^{-1}(1-t)^{n+1}\prod_{i=2}^{n-1}(1-it)$.
\end{theorem}

\begin{proof}\ \\
That $\begin{pmatrix} g_{0,i}\\ g_{1,i}\end{pmatrix}$ is linear independent to $\begin{pmatrix} g_{0,j}\\ g_{1,j}\end{pmatrix}$ for $j\in\{1,\hdots,i-1\}$ implies that the entries of $G(z)$ are coprime and hence each $(n,1,1)$ MDP convolutional code is non-catastrophic. Thus, to get the conditional probability, one just has to divide the formula of the preceding theorem by the probability of non-catastrophicity, which is $1-t^{n-1}$; see \cite{ich}.
\end{proof}

\subsection{$(n,n-1,1)$ convolutional codes for $n\geq 3$}

As the $(n,n-1,1)$ MDP convolutional codes are dual to the $(n,1,1)$ MDP convolutional codes treated in the previous subsection, we easily get all $(n,n-1,1)$ MDP convolutional codes and know that they exist if and only if $|\mathbb F|\geq n$. For the construction, we just replace the conditions on the generator matrix $G$ from the preceding subsection by the same conditions on the parity-check matrix $H$. The following theorem considers $(n,n-1,1)$ reverse and complete MDP convolutional codes.

\begin{theorem}\ \\
The number of $(n,n-1,1)$ reverse and complete MDP convolutional codes is both $(|\mathbb F|-1)^{n+1}\prod_{i=2}^{n}(|\mathbb F|-i)$. Hence, the minimal field size for which an $(n,n-1,1)$ reverse or complete MDP convolutional code could exist is $n+1$.
\end{theorem}

\begin{proof}\ \\
To get reverse or complete MDP convolutional codes with these parameters, one has the additional condition that all entries of $h_1$ have to be nonzero. Thus, there are $(|\mathbb F|-1)^{n+1}\prod_{i=2}^{n}(|\mathbb F|-i)$ such convolutional codes, which could be constructed with the same technique as mentioned before.
\end{proof}

\section{Sufficient field size and probability for MDP convolutional codes with arbitrary parameters}

\subsection{Sufficient field size}

The goal of this subsection is to estimate what field size one needs such that it is possible to construct an MDP convolutional code with given but arbitrary parameters $n$, $k$ and $\delta$. 

\begin{theorem}\ \\
Let $g$ be the polynomial that is formed by the product of all not trivially zero fullsize minors of $\mathcal{H}_L$ and has the entries of the coefficient matrices of $H$ as variables. Then, an $(n,k,\delta)$ MDP convolutional code exists if $|\mathbb F|>\deg(g)$.
\end{theorem}

\begin{proof}\ \\
It is sufficient to show the existence of an MDP convolutional code with generic row degrees, i.e. $\nu=\lceil\frac{\delta}{n-k}\rceil$ and consider only matrices $H$ of the form of Remark \ref{kh}. This means $H_{\nu}$ and $H_{\nu-1}$ are of the forms $H_{\nu}=\left(\begin{array}{ccc}
I_{\delta-(n-k)(\nu-1)} & 0 & \ast \\ 
0 & 0_{(n-k)\nu-\delta} & 0
\end{array} \right)$ and $H_{\nu-1}=\left(\begin{array}{ccc}
\ast & 0 & \ast \\ 
0_{((n-k)\nu-\delta)\times(\delta-(n-k)(\nu-1))} & I_{(n-k)\nu-\delta} & \ast
\end{array} \right)$, respectively. In this way, one can ensure that the generated code has really the given degree $\delta$.

According to Theorem \ref{sz} (a), a polynomial $g$ over $\mathbb F$ with $m$ variables (and $\deg(g)\geq 0$) has at most $\deg(g)\cdot|\mathbb F|^{m-1}$ zeros. Altogether, there are $|\mathbb F|^m$ tuples of points. Therefore, for having at least one of them being not a zero, it is sufficient that  $|\mathbb F|^m>\deg(g)\cdot|\mathbb F|^{m-1}$, i.e. $|\mathbb F|>\deg(g)$.

 We apply this result to the polynomial $g$ formed by the product of all not trivially zero fullsize minors of $\mathfrak{H}_L$.

 Since some of the entries of $H_{\nu}$ and $H_{\nu-1}$ are fixed zeros or ones, one has less variables than the number of entries of the coefficient matrices of $H$ but this has no influence on the result (note that if $\nu>L$, i.e. $\nu=L+1$, and $\lfloor\frac{\delta}{k}\rfloor=0$, for which case we give a better bound in a later section, all entries of $H_{\nu}$ do not occur in the polynomial on which we apply Schwartz-Zippel). What influences Schwartz-Zippel is not the number of variables but the degree of the polynomial $g$. This degree is in all cases at most $(L+1)(n-k)$ times the number of not trivially zero fullsize minors of $\mathfrak{H}_L$.
\end{proof}

%
%
%

It remains to estimate the degree of the polynomial $g$ from the preceding theorem to get an explicite bound for the field size.

\begin{theorem}\label{M}\ \\
If $|\mathbb F|>\min\{M_1,M_2,M_3\}$ with
\begin{align*}
M_1&:=(L+1)(n-k)\binom{(L+1)n}{(L+1)(n-k)}\\
M_2&:=(L+1)(n-k)\binom{n}{n-k}\binom{n+k}{n-k}\cdots\binom{n+Lk}{n-k}\\
M_3&:=(L+1)(n-k)\sum_{i=n-k}^{n}\binom{i-1}{n-k-1}\binom{2n-i}{n-k}\cdots\binom{2n+(L-1)k-i}{n-k}
\end{align*}
then there exists an $(n,k,\delta)$ MDP convolutional code over $\mathbb F$.
\end{theorem}

\begin{proof}\ \\
To show that $|\mathbb F|>M_i$ for some $i\in\{1,2,3\}$ is sufficient, one has to show that the factor after $(L+1)(n-k)$ in the formulas is an upper bound for the number of not trivially zero fullsize minors of $\mathcal{H}_L$. For $M_1$ this is clear because there we use just the formula for all fullsize minors.\\
For $M_2$, we use the condition that we have to choose $n-k$ columns from the first $n$ columns, then $n-k$ columns from the first $2n$ columns without the $n-k$ columns we have already chosen and so on until we end up with choosing $n-k$ columns from $(L+1)n$ columns without the $L(n-k)$ columns we have already chosen.\\
For $M_3$, we denote by $i\in\{n-k,\hdots,n\}$ the index of the $(n-k)$-th column we choose. Thus, one has to choose $n-k-1$ columns with smaller index than $i$, i.e. out of the first $i-1$ columns of $\mathcal{H}_L$. After that, one proceeds like for $M_2$, i.e. next one has to choose $n-k$ columns out of $2n$ but not the first $i$, then $n-k$ out of $3n$ without the first $i$ and without the $n-k$ chosen in the preceding step and so on.
\end{proof}

\begin{remark}\ \\
It depends one the parameters of the code, which of the bounds is best. In the following, we give some examples:
\begin{enumerate}
\item Case $L=0$: $M_1=M_2=M_3$ (in this case there are no trivially zero minors)
\item Case $k=1$\\
$M_3=M_2\cdot\left(\frac{1}{n}+\frac{n-1}{\binom{n+L}{n-1}}\right)$
\begin{itemize}
\item[(a)] $L\geq 1$: $M_3<M_2$
\item[(b)] 
$L=1\ (\Rightarrow \delta=1, n\geq 3)$: $M_3=(3n^2-n)(n-1)<M_1=(4n^2-2n)(n-1)<M_2=(n^3+n^2)(n-1)$
\item[(c)] 
$(2,1,1)$: $M_3=18\cdot 3<M_1=20\cdot 3<M_2=24\cdot 3$
\item[(d)]
$(2,1,\delta)$ with $\delta\geq 2$: $M_1<M_3<M_2$\\
It holds $L\geq 4$, which implies $M_3<M_2$ according to (a), and for $L=4$, $M_1=252\cdot 5<M_3=480\cdot 5<M_2=720\cdot 5$. Moreover, $M_3$ is increasing more than $M_1$, when $L$ increases (to $L+1$). This is true since $M_1/(L+1)=\binom{2(L+1)}{L+1}$ increases with factor $\frac{(2L+3)(2L+4)}{(L+2)(L+2)}<4$ and $M_3/(L+1)=(L+2)!\cdot (\frac{1}{2}+\frac{1}{L+2})$ increases with factor $(L+2)\frac{L^2+7L+10}{L^2+7L+12}>5$ for $L\geq 4$.
\end{itemize}
\end{enumerate}

\end{remark}

\subsection{Probability}

In this subsection, we want to compute the probability that a non-catastrophic convolutional code with arbitrary parameters is MDP. Therefore, we assign to each code the unique parity-check matrix from Remark \ref{kh}. This is possible since permutation of the columns of the parity-check matrix does not influence the MDP property.
With these definitions/settings, one gets the following theorem:

\begin{theorem}\ \\
Let $\mathbb F$ be finite with cardinality $|\mathbb F|=t^{-1}$. If $|\mathbb F|>\min\{M_1,M_2,M_3\}$, the probability for an MDP convolutional code is lower bounded by\\
(i) $1-\frac{(L+1)(n-k)\binom{(L+1)n}{(L+1)(n-k)}\cdot t}{1-t^k+O(t^{k+1})}$\\
(ii) $1-\frac{\left((L+1)(n-k)\binom{n}{n-k}\binom{n+k}{n-k}\cdots\binom{n+Lk}{n-k}\right)\cdot t}{1-t^k+O(t^{k+1})}$\\
(iii) $1-\frac{\left((L+1)(n-k)\sum_{i=n-k}^{n}\binom{i-1}{n-k-1}\binom{2n-i}{n-k}\cdots\binom{2n+(L-1)k-i}{n-k}\right)\cdot t}{1-t^k+O(t^{k+1})}$
\end{theorem}

\begin{proof}\ \\
For MDP (in contrast to complete MDP) it is not necessary that $H$ has generic row degrees. Therefore, one has to make the following considerations for all possible values of the row degrees.
However, we will see that this does not matter. 
 
Again, we assume that $H=[P\ Q]$ has the form of Remark \ref{kh}.
If the row degrees of $Q$ are fixed, one knows for each entry of $H$ either its degree or an upper bound on its degree. Hence, when considering the entries of the coefficient matrices of $H$ as variables, we know how many variables we have and could apply Theorem \ref{sz} (b) to the polynomial $g$ that is formed by the product of the non-trivially fullsize minors of $\mathcal{H}_L$. Note that this polynomial is not the zero polynomial (since the existence of MDP convolutional codes has been shown for $|\mathbb F|>\min\{M_1,M_2,M_3\}$). 

It has already been shown that $M_1$, $M_2$ and $M_3$ are upper bounds for $\deg(g)$.

By the Schwartz-Zippel, the probability that the variables do not fulfill the condition for  MDP is upper bounded by $\deg(g)\cdot t$.\\
One has to consider conditional probability with the condition that $Q$ and $P$ are left coprime. Therefore, the overall probability is upper bounded by the absolute probability divided by the probability of the condition, which is $1-t^k+O(t^{k+1})$; see \cite{ich}.
\end{proof}

\section{Sufficient field size and probability for complete MDP convolutional codes}

In this section, we want to do the same considerations for complete MDP convolutional codes that were done for MDP convolutional codes in the preceding section.

\subsection{Sufficient field size}

\begin{theorem}\ \\
Let $f$ be the polynomial that is formed by the product of all not trivially zero fullsize minors of $\mathfrak{H}$ and has the entries of the coefficient matrices of $H$ as variables. Then, for $(n-k)\mid\delta$, an $(n,k,\delta)$ complete MDP convolutional code exists if $|\mathbb F|>\deg(f)$.
\end{theorem}

\begin{proof}\ \\
One uses Schwartz-Zippel and proceeds completely analogous to the preceding subsection.
\end{proof}

Again, we have to estimate the degree of the polynomial $f$ from the preceding theorem to get an explicite bound for the field size.

\begin{theorem}\ \\
If $(n-k)\mid\delta$ and $|\mathbb F|>\min\{N_1,N_2\}$ with
\begin{align}
N_1&:=(L+1)(n-k)\binom{(L+1+\frac{\delta}{n-k})n}{(L+1)(n-k)}\\
N_2&:=(L+1)(n-k)(\frac{\delta n}{n-k} +k+1)^{(n-k)(L+1)},
\end{align}
then there exists an $(n,k,\delta)$ complete MDP convolutional code over $\mathbb F$.
\end{theorem}

\begin{proof}\ \\
Each fullsize minor of $\mathfrak{H}$ is a polynomial of degree $(L+1)(n-k)$.
Moreover, the number of not-trivially zero fullsize minors of $\mathfrak{H}$ is upper bounded by $\binom{L+1+\frac{\delta}{n-k}}{(L+1)(n-k)}$, which is the number of all fullsize minors, as well as by $(\frac{\delta n}{n-k} +k+1)^{(L+1)(n-k)}$ since the index of each chosen column has to lie in an interval with $\frac{\delta n}{n-k} +k+1$ elements (see Lemma \ref{index}).
\end{proof}

\begin{remark}\ \\
(i) It depends on the parameters of the code, which of the two bounds $N_1$ or $N_2$ is better, i.e. smaller. For example for $k=n-1$, the second bound is better for $n=2$, for $n=3$ the bounds are identical, and for $n\geq 4$ the first bound is better.\\
(ii)
For complete MDP, one has $\nu=\frac{\delta}{n-k}$ and hence $L\geq\nu\geq 1$, which implies $\binom{\nu n+k}{\lfloor 1/2(\nu n+k)\rfloor}^{(n-k)(L+1)}\cdot((n-k)(L+1))^{1/2(n-k)(L+1)}\geq \binom{\nu n+k}{\lfloor 1/2(\nu n+k)\rfloor}^{(n-k)(L+1)}\cdot(n-k)(L+1)\geq (L+1)(n-k)(\frac{\delta n}{n-k} +k+1)^{(n-k)(L+1)}$ if $(n,k,\delta)\neq(2,1,1)$. This shows that - unless $(n,k,\delta)=(2,1,1)$ - the bound on the field size presented here is better than the bound obtained by the construction in \cite{cmdp}, which is clearly very weak (which is due to the fact that it provides a general construction) but up to now there did not exist better bounds. For $(2,1,1)$, we have already seen that the minimal possible field size is $3$, i.e. much smaller than all these bounds.
\end{remark}

\subsection{Probability}

We want to compute the probability that a non-catastrophic convolutional code with $(n-k)\mid\delta$ and generic row degrees $\nu=\frac{\delta}{n-k}$ is complete MDP. Therefore, we assign again to each code the unique parity-check matrix from Remark \ref{kh}. This is possible since permutation of the columns of the parity-check matrix does not influence the property to be complete MDP.
%
With these definitions/settings, one gets the following theorem:

\begin{theorem}\ \\
If $|\mathbb F|>\min\{N_1,N_2\}$, the probability for a complete MDP convolutional code is at least
$$\max\left\{1-\frac{(L+1)(n-k)\binom{(L+1+\frac{\delta}{n-k})n}{(L+1)(n-k)}\cdot t}{1-t^k+O(t^{k+1})},1-\frac{(L+1)(n-k)(\frac{\delta n}{n-k} +k+1)^{(n-k)(L+1)}\cdot t}{1-t^k+O(t^{k+1})} \right\}$$
\end{theorem}

\begin{proof}\ \\
The proof is completely analogue to the proof for the probability of MDP convolutional codes.
\end{proof}

\section{Sufficient field size for MDP convolutional codes with $\delta<\max\{k,n-k\}$}

In this section, we show a better bound on the necessary field size for MDP convolutional codes for the case that $\delta<\max\{k,n-k\}$. Because of duality arguments, we mainly have just to solve the case $\delta<k$.

\subsection{The case $\delta<k$}

To derive an upper bound for the required field size, one could assume that $H$ has generic row degrees since the existence of an MDP convolutional code with generic row degrees over $\mathbb F$ obviously implies the existence of an MDP convolutional code over $\mathbb F$. 
The genericity of the row degrees implies $\nu=\lceil\frac{\delta}{n-k}\rceil$ and therefore $L= \lfloor\frac{\delta}{n-k}\rfloor\leq\nu$. If $(n-k)\mid \delta$, i.e. $L=\nu$, all row degrees of $H$ are equal to $\nu$ and hence $H_L=H_{\nu}$ does not contain fixed zeros. If $(n-k)\nmid\delta$, i.e. $L=\nu-1$ and all row degrees of $H$ are either equal to $\nu-1$ or equal to $\nu$, $H_L$ does not contain fixed zeros, too.

\begin{theorem}\label{bb}\ \\
There exists an MDP convolutional code with $\delta<k$ over $\mathbb F$ if either
\begin{enumerate}
\item
$|\mathbb F|>\binom{(L+1)n-1}{(L+1)(n-k)-1}$ \textbf{or}
\item
in the case $L\geq 1$, $|\mathbb F|>S(n,k,\delta)$ with
\begin{align*}S(n,k,\delta)&:=\sum_{j=n-k+1}^{(n-k)L-1}\binom{n-1}{j-1}\binom{\lfloor\frac{j}{n-k}\rfloor n}{(\lfloor\frac{j}{n-k}\rfloor+1)(n-k)-j}\binom{(\lfloor\frac{j}{n-k}\rfloor+1)n}{n-k}\cdots\binom{Ln}{n-k}+\\
&+\sum_{j=\max\{(n-k)L,n-k+1\}}^{(n-k)(L+1)-1}\binom{n-1}{j-1}\binom{Ln}{(L+1)(n-k)-j}+\binom{n-1}{(L+1)(n-k)-1}
\end{align*}
\end{enumerate}
\end{theorem}

\begin{proof}\ \\
For $y\in\{1,\hdots,(L+1)(n-k)\}$, define $\mathcal{H}_L^{(y)}$ as the matrix consisting of the first $y$ rows of $\mathcal{H}_L=\left[\begin{matrix} H_0 & & 0\\ \vdots & \ddots & \\ H_{L} & \cdots & H_0 \end{matrix}\right]$.

We prove via induction with respect to $y$ that if $\mathbb F$ fullfilles condition $1$ or $2$, then it is possible to find values for $H_0,\hdots,H_L$ over $\mathbb F$ such that every fullsize minor of $\mathcal{H}_L^{(y)}$ that is not trivially zero is nonzero.

For $y=1$, all entries in the first row of $H_0$ have to be nonzero, what is possible if $|\mathbb F|>1$, which is implied by both condition 1 and condition 2 (but true for any field anyway).

Assume that the statement is valid for $1,\hdots,y$.
For the step to $y+1$, consider the last row of $\mathcal{H}_L^{(y+1)}$. First, we show that for each $i\in\{1,\hdots,n\}$, if all entries of $\mathcal{H}_L^{(y+1)}$ but the $i$-th entry of the last row of $\mathcal{H}_L^{(y+1)}$, named by $\mathcal{H}_{L,i}$, are fixed (such that the statement is valid for $1,\hdots,y$), there is a possibility to choose $\mathcal{H}_{L,i}$ from $\overline{\mathbb F}$ such that the statement is valid for $y+1$.

To do this, we consider all not trivially zero fullsize minors of $\mathcal{H}_L^{(y+1)}$ that contain the $i$-th column of this matrix.

For each of these minors, one has to show that it is possible to choose   $\mathcal{H}_{L,i}$ such that the minor is nonzero.

 Denote by $M\in\mathbb F^{(y+1)\times(y+1)}$ the submatrix of $\mathcal{H}_L^{(y+1)}$ that corresponds to the considered fullsize minor and let $\hat{M}$ be constructed out of $M$ by deleting the row and the column that contain $\mathcal{H}_{L,i}$.

Hence, in the case $\det(\hat{M})=0$, one has to show $\det(M)\neq 0$, independent of the choice of $\mathcal{H}_{L,i}$.

Since $\hat{M}$ is a fullsize minor of $\mathcal{H}_L^{(y)}$, it follows by induction that it has to be trivially zero. Because of the structure of $\mathcal{H}_L$ this implies that there exists $s\in\{1,\hdots,L\}$ such that column $s(n-k)$ of $\hat{M}$ is a column of $\mathcal{H}_L^{(y)}$ with index at least $sn+1$. Moreover, it follows that this column is column $s(n-k)+1$ of $M$ and its first $s(n-k)$ entries are zeros since it is not from the first $s$ blocks of $\mathcal{H}_L$.

Consequently, $M$ is of the following form: $\left[\begin{matrix} A & 0_{s(n-k)\times (y+1-s(n-k))}\\ \ast & B \end{matrix}\right]\in\mathbb F^{(y+1)\times(y+1)}$. Hence $A$ and $B$ are square matrices with $\det(M)=\det(A)\cdot\det(B)$. Moreover, $A$ and $B$ are fullsize submatrices of $\left[\begin{matrix} H_0 & & 0\\ \vdots & \ddots & \\ H_{s-1} & \cdots & H_0 \end{matrix}\right]$ and $\left[\begin{matrix} H_0 & & 0\\ \vdots & \ddots & \\ H_{L-s} & \cdots & H_0 \end{matrix}\right]$, respectively.\\
Since the columns of $M$ are chosen such that $\det(M)$ is not trivially zero, $\det(A)$ and $\det(B)$ are not trivially zero, too. By induction it follows that $\det(A)$ and $\det(B)$ are nonzero and therefore also $\det(M)$ is nonzero.

To show that one can find such $\mathcal{H}_{L,i}$ over $\mathbb F$ if condition 1 or 2 is fulfilled, we count the maximum number of values that have to be excluded for $\mathcal{H}_{L,i}$, where without restriction, one could assume $i=n$ (note that $\hat{M}$ is independent of $H_{L,i}$ as well as $\det(M)$ in the case $\det(\hat{M})=0$).
This number is upper bounded by the number of not trivially zero fullsize minors of $\left[\begin{matrix} H_0 & & 0\\ \vdots & \ddots & \\ H_{L} & \cdots & H_0 \end{matrix}\right]$ with $r_j=n$ for some $j\in\{1,\hdots,(L+1)(n-k)\}$ since all these minors are at most linear in $H_{L,n}$. For the bound of condition 1, we just count the number of all fullsize minors with $r_j=n$ no matter if they are trivially zero or not. Surely, it is sufficient if $\mathbb F$ has more elements as the number of these minors.\\
For condition 2, which takes into account that some minors are trivially zero, one could assume $L\neq 0$ and neglect the case $j\leq n-k$. Since $r_{j+1}>n$, the minor would be trivially zero if $j<n-k$. If $j=n-k$, one chooses exactly $n-k$ columns from the first block of $\mathcal{H}_L$ and the minor is nonzero if and only if the corresponding $n-k$ columns of $H_0$ are linearly independent and the matrix (for the minor) without the first $n-k$ columns and rows has full rank. But these conditions are independent of $H_L$ and hence, do not lead to values for $H_{L,n}$ that have to be excluded.

For $j\geq n-k+1$, there are at most $\binom{n-1}{j-1}$ possiblitities to choose $r_1,\hdots,r_{j-1}$ since one has the condition $1\leq r_1<r_2<\cdots<r_{j-1}<r_j=n$. For $r_l$ with $l>j$, one has to consider the condition $r_{s(n-k)}\leq sn$ for $s=1,\hdots, L$. $r_j=n$ implies that this condition is already fulfilled for $s=1,\hdots,\lfloor\frac{j}{n-k}\rfloor$. To fulfil this condition for $s=\lfloor\frac{j}{n-k}\rfloor+1$, we need to choose $(\lfloor\frac{j}{n-k}\rfloor+1)(n-k)-j$ columns from the first $(\lfloor\frac{j}{n-k}\rfloor+1)n$ columns but not from the first $n$ columns of $\mathfrak{H}_L$, i.e. we have to choose $(\lfloor\frac{j}{n-k}\rfloor+1)(n-k)-j$ columns out of $(\lfloor\frac{j}{n-k}\rfloor)n$ columns. For $s\geq \lfloor\frac{j}{n-k}\rfloor+2$, we have to choose $n-k$ columns out of at most $sn-n$ columns. Summing over all possible values for $j$, one gets the formula from condition 2.\\

To ensure that the degree of the code is equal to $\delta=\nu_1+\cdots+\nu_{n-k}$, one has to ensure the $H$ is row proper, i.e. that the highest row rank coefficient matrix is invertible. This is true if the first $\delta-(n-k)(\nu-1)$ rows of $H_{\nu}$ and the last $(n-k)\nu-\delta$ rows of $H_{\nu-1}$ are linearly independent. When choosing the entries of $\mathcal{H}_L$ row by row as done in this proof, the number of values that has to be excluded for each entry of a coefficient matrix increases in each step. Moreover the condition that the highest row degree coefficient matrix is invertible, i.e. that the above mentioned rows are linearly independent, could be fulfilled by the first $n-k$ columns of $H_{\nu-1}$ and $H_{\nu}$. Therefore, one has no additional condition on $H_{L,n}$ because of that and thus, no additional value has to be excluded.
(Note that for $(n-k)\nmid\delta$, i.e. $\nu>L$, $H_{\nu}$ is not contained in $\mathcal{H}_L$ and the only thing that has to be regarded when choosing the values for $H_{\nu}$ is that $H$ has to be column proper).
\end{proof}

It would be possible to adopt condition 2 such that is valid also for $L=0$ (in principle, the difference would be that then, one had to take $j=n-k$ as lower bound for the first sum since the case $j=n-k$ cannot be neglected for $L=0$). But for $L=0$ one has no trivially zero fullsize minors in $\mathcal{H}_L=H_0$ and therefore, it would equal the bound of condition 1 for $L=0$, anyway.

\begin{corollary}\label{bb1}\ \\
Bound 1 of the preceding theorem can be upper bounded by the following expression, which is independent of $k$:
$$\binom{(L+1)n-1}{(L+1)(n-k)-1}\leq \binom{(L+1)n-1}{n-2}.$$
\end{corollary}

\begin{proof}\ \\
Per definition, $L=\lfloor\frac{\delta}{k}\rfloor + \lfloor\frac{\delta}{n-k}\rfloor =\lfloor\frac{\delta}{n-k}\rfloor$ as $k>\delta$. Hence $L\leq \frac{\delta}{n-k}$, i.e. $L(n-k)\leq\delta$. It follows $(L+1)(n-k)-1=L(n-k)+n-k-1\leq \delta +n-k-1\leq n-2$ since $k>\delta$. Moreover for $L\geq 1$, one has $2(n-2)+1\leq 2n-3< (L+1)n-1$ and thus, $\binom{(L+1)n-1}{(L+1)(n-k)-1}\leq \binom{(L+1)n-1}{n-2}$. 
\end{proof}

\begin{remark}\label{coinc}\ \\
For $k=n-1$, i.e. $(n,n-1,\delta)$ convolutional codes with $\delta\leq n-2$, one has $L=\delta$ and the bound of condition 2 equals 
\begin{align*}
&\sum_{j=2}^{\delta-1}\binom{n-1}{j-1}\binom{jn}{(j+1)-j}\binom{(j+1)n}{1}\cdots\binom{\delta n}{1}+\binom{n-1}{\delta-1}\binom{\delta n}{1}+\binom{n-1}{\delta}=\\
&=\sum_{j=2}^{\delta+1}\binom{n-1}{j-1}\cdot n^{\delta+1-j}\cdot \frac{\delta!}{(j-1)!}=\delta!\cdot n^{\delta}\sum_{j=1}^{\delta}\binom{n-1}{j}\cdot \frac{n^{-j}}{j!}<\delta!\cdot n^{\delta}(e-1)
\end{align*}
as
\begin{align*}
\sum_{j=1}^{\delta}\binom{n-1}{j}\cdot n^{-j}<\sum_{j=1}^{n-1}\binom{n-1}{j}\cdot n^{-j}=(1+1/n)^{n-1}-1<e-1.
\end{align*}

Setting also $\delta=1$ (for which one needs $n\geq\delta+2=3$), one gets $n-1$. This implies that the bound is sharp in that case; see Section 3.
\end{remark}

\subsection{The case $\delta<n-k$}

For this case, we could use again that the MDP property is invariant under duality.

\begin{theorem}\ \\
If $\delta<n-k$ and $|\mathbb F|>S(n,n-k,\delta)$ or $|\mathbb F|>\binom{(L+1)n-1}{(L+1)k-1}$, there exists an $(n,k,\delta)$ MDP convolutional code over $\mathbb F$.
\end{theorem}

\begin{proof}\ \\
The result follows from the preceding subsection and Theorem \ref{dual}.
\end{proof}

\begin{corollary}\label{bb2}\ \\
Analogous to the preceding subsection, one gets for $\delta<n-k$ that
$$\binom{(L+1)n-1}{(L+1)k-1}\leq \binom{(L+1)n-1}{n-2}.$$
\end{corollary}
\subsection{The case $\delta<\min\{k,n-k\}$}

In this subsection, we consider the case that the conditions of both preceding subsections are fulfilled, resulting in $L=0$.

\begin{theorem}\label{L0}\ \\
If $\delta<\min\{k,n-k\}$ and $|\mathbb F|>\min\left\{\binom{n-1}{n-k-1},\binom{n-1}{k-1}\right\}$, there exists an $(n,k,\delta)$ MDP convolutional code over $\mathbb F$.
\end{theorem}

\begin{proof}\ \\
The result follows from the results of the preceding subsections.
\end{proof}

\begin{remark}\ \\
(i)
For $\delta<\min\{k,n-k\}$, i.e. $L=0$, one has $\mathfrak{H}_L=H_0$ and the corresponding code is MDP if and only if $\begin{pmatrix}I_k\\ H_0 \end{pmatrix}$ is the generator matrix of an $[n,k]$ MDS block code. Therefore, the bound of the preceding theorem is a bound for the necessary field size for the existence of an $[n,k]$ MDS block code.\\
(ii)
$n-k>\delta$ implies that $n-k$ cannot divide $\delta$ and therefore, there exists no complete MDP convolutional code with these parameters. However, one could show that for a reverse MDP convolutional code with $L=0$, the bound for MDP codes is also sufficient: $G_1,\hdots, G_{\mu}$ do not influence the property to be MDP. Thus, one could choose them arbitrary without affecting the MDP property. If one chooses column $j$ of $G_{\delta_j}$ equal to column $j$ of $G_0$, one gets $\overline{G}_0=G_0$ and therefore, the code is also reverse MDP.\\
(iii)
For an $(n,n-2,1)$ convolutional code, the bound of the preceding theorem is equal to $n-1$. One could easily see that the construction $H_0=\left(\begin{array}{cccc}
1 & 1 & \hdots & 1 \\ 
0 & 1 & \hdots & n-1
\end{array} \right)$ reaches this bound. But as we will see in the next section, constructions over fields of smaller size are possible.
\end{remark}

\section{Comparison of bounds}

The aim of this section is to show that in nearly all cases, the bounds on the necessary field size for the existence of MDP convolutional codes presented in this paper could improve all bounds that were proven before. Therefore, we start with recalling which bounds already existed. The following theorem gives the only bound up to now that is valid for all code parameters.

\begin{theorem}\cite{b}\ \\
Let $B_{\gamma}:=\frac{1}{2}\left(\frac{1}{\gamma}\binom{2(\gamma-1)}{\gamma-1}+\binom{\gamma-1}{\lfloor\frac{\gamma-1}{2}\rfloor}\right)$
and $\mathbb F$ be a finite field with $|\mathbb F|>B_{\gamma}$ . Then, there exists a $\gamma\times\gamma$ superregular Toeplitz matrix over $\mathbb F$.

Let $r$ be the remainder of $\delta$ on division by $n-k$. Let $\mathbb F$ be a finite field with $|\mathbb F| > B_{(L+1)(n-1)}$ or $|\mathbb F| > B_{(L+1)(n-1)+k+r-1}$ as $r = 0$ or $r\neq 0$, respectively.\\
Then, an $(n, k, \delta)$ MDP convolutional code exists over $\mathbb F$.
\end{theorem}

This theorem as well as the following conjecture use in the same way square superregualr Toeplitz matrices to construct MDP convolutional codes.

\begin{conjecture}\cite{b},\cite{strongly}\ \\
For $\gamma\geq 5$, there is a $\gamma\times\gamma$ superregular Toeplitz matrix over $\mathbb F_{2^{\gamma-2}}$.\\
Let $r$ be the remainder of $\delta$ on division by $n-k$. Let $\mathbb F$ be a finite field with $|\mathbb F|\geq2^{(L+1)(n-1)-2}$ or $|\mathbb F|\geq 2^{(L+1)(n-1)+k+r-3}$ as $r = 0$ or $r\neq 0$, respectively.\\
Then, an $(n, k, \delta)$ MDP convolutional code exists over $\mathbb F$.
\end{conjecture}

The preceding conjecture would yield a better bound than $B_{\gamma}$ (see \cite{b}) but would not be sharp as the following table from \cite{b} shows:\\

\begin{tabular}{|c|c|c|c|c|c|c|c|c|}
\hline 
Size of superregular Toeplitz matrix & 3 & 4 & 5 & 6 & 7 & 8 & 9 & 10 \\ 
\hline 
Minimum required field size & 3 & 5 & 7 & 11 & 17 & 31 & 59 & $\leq 127$ \\ 
\hline 
\end{tabular}\\

For small parameters, the preceding table provides the exact necessary field size such that a superregular Toeplitz matrix exists. We will see later that for several of the parameters covered by this table, it is possible to derive MDP convolutional codes over fields of smaller size when using other constructions than via superregular Toeplitz matrices.\\

For a very special choice of parameters, the minimum required field size for MDP convolutional codes has been obtained in \cite{bar}, where the authors also provide a corresponding construction of such codes.

\begin{theorem}\cite{bar}\ \\
For $m>2$ an $(2^{m-1},2^{m-1}-1,2)$ MDP convolutional code exists if and only if $|\mathbb F|\geq 2^m$.
\end{theorem}

Applying Theorem \ref{bb} to $(n,n-1,2)$ convolutional codes, one gets that  $|F|>(n-1)(2,5n-1)$ is sufficient. Clearly, for the case that $n$ is an exponent of $2$, this bound is much weaker than the bound from \cite{bar}. But in turn, the bound from Theorem \ref{bb} works for general $n$.

Now, we want to compare the new bounds of this paper with the already existing bounds (but the bound in \cite{bar} since that bound is optimal anyway) and start with the case $L=0$.

\subsection{Comparison of bounds for $L=0$}

\begin{theorem}\ \\
The bound for $L=0$ from Theorem \ref{L0} is better than $B_{\gamma}$, the conjecture and and than using exact values for superregularity for small matrices.
\end{theorem}

\begin{proof}\ \\
$L=0$ implies $n-k>\delta$, i.e. $n-k$ does not divide $\delta$ and thus $r\geq 1$. Hence,
$\binom{n-1}{n-k-1}<2^{n-1}\leq 2^{n-1+k+r-1-2}$ shows that the new bound is better than the conjecture, which implies than it is better than using the bound $B_{\gamma}$. The first inequality follows from the Stirling formula and the second since $k>\delta\geq 1$.\\
Exact values for the existence of superregular matrices are only known if $n+k+r-2\leq 10$. Computing all cases in which this inequation as well as $\min\{n-k,k\}>\delta$ is fullfilled, which implies $r=\delta$, one sees that the new bound is always smaller.
\end{proof}

We want to show with some small examples that there are cases in which $L=0$ and our bound is not optimal (even if it is the best of the existing bounds).

\begin{example}
\begin{enumerate}
\item According to Theorem \ref{L0}, an $(4,2,1)$ MDP convolutional code exists if $|\mathbb F|\geq 4$. But $\mathcal{H}_L=H_0:=\left(\begin{array}{cccc}
1 & 0 & 1 & 2 \\ 
0 & 1 & 1 & 1
\end{array} \right)$ yields an MDP convolutional code with these parameters over $\mathbb F_3$.
\item 
According to Theorem \ref{L0}, $(6,3,1)$ and $(6,3,2)$ MDP convolutional codes exist if $|\mathbb F|\geq 11$. But $\mathcal{H}_L=H_0:=\left(\begin{array}{cccccc}
1 & 0 & 0 & 1 & 1 & 2 \\ 
0 & 1 & 0 & 1 & 2 & 1 \\ 
0 & 0 & 1 & 1 & 3 & 3
\end{array}\right)$ yields MDP convolutional codes with these parameters over $\mathbb F_5$.
\end{enumerate}
\end{example}

\subsection{Comparison of bounds for $L\geq 1$ and $\delta<\max\{k,n-k\}$}


\begin{theorem}\ \\
(i)
Bound 2 from Theorem \ref{bb} is always better than $B_{\gamma}$ and the conjecture. It is also better than using exact values for superregularity for small matrices - except for $(5,3,2)$ codes, where bound 2 yields the existence of MDP codes if $|\mathbb F|> 34$ and using the table for existence of superregular matrices, one gets that $|\mathbb F|\geq 31$ is sufficient.

(ii)
Bound 1 from Theorem \ref{bb} is for nearly all cases the next best bound after bound 2, exceptions are only $(15,8,7)$, $(13,7,6)$, $(11,6,5)$, $(9,5,4)$, $(7,4,3)$, $(6,4,2)$, $(5,3,2)$, $(3,2,1)$. For $(3,2,1)$, using $B_{\gamma}$, which is here identical with the exact minimal value for superregularity, yields a bound between 2 and 1. For $(6,4,2)$ using the exact value for superregularity yields a bound between 2 and 1. In the other cases, the conjecture yields a bound between 2 and 1 but the conjecture has not been proven. 
\end{theorem}

\begin{proof}\ \\
\textbf{Step 1: Bound 2 is always better than bound 1}\\
In bound 2, for each $j\in\{n-k+1,\hdots,(n-k)(L+1)\}$, one chooses altogether $(L+1)(n-k)-1$ elements from at most $(L+1)n-1$ elements (for $n-k+1\leq j<L(n-k)$, this is true since $(L-\lfloor\frac{j}{n-k}\rfloor n)(n-k)+(\lfloor\frac{j}{n-k}\rfloor+1)(n-k)+1=(L+1)(n-k)-1$ and $(L-\lfloor\frac{j}{n-k}\rfloor+1)n+n-1\leq (L+2-\lfloor\frac{n-k+1}{n-k}\rfloor)n-1\leq (L+1)n-1$). Hereby, one has certain conditions for this choice. Bound 2 computes all possibilities without any restrictions to choose $(L+1)(n-k)-1$ elements from $(L+1)n-1$ elements and is therefore larger.

\textbf{Step 2: We use the upper bound for bound 1 from Corollary \ref{bb1} and Corollary \ref{bb2}}\\
According to Corollary \ref{bb1} and Corollary \ref{bb2}, it is sufficient to show that $\binom{(L+1)n-1}{n-2}$ is always smaller than the conjecture.

\textbf{Step 3: For $L=2$, bound 1 is better than the conjecture (and therefore, also better than $B_{\gamma}$)}\\
Bound 1 is at most $\binom{3n-1}{n-2}<\frac{2^{3n-1}}{\sqrt{\pi}\sqrt{\lfloor\frac{3n-1}{2}\rfloor}}$ by the Stirling formula and the conjecture is equal to at least $2^{3n-5}$. For $n\leq 6$, one could compute directly that $\binom{3n-1}{n-2}$ is smaller than $2^{3n-5}$. For $n=7$, one could compute that $\frac{2^{3n-1}}{\sqrt{\pi}\sqrt{\lfloor\frac{3n-1}{2}\rfloor}}$ is smaller than $2^{3n-5}$. Since $2^{3n-5}$ is increasing more rapidly than $\frac{2^{3n-1}}{\sqrt{\pi}\sqrt{\lfloor\frac{3n-1}{2}\rfloor}}$ when $n$ increases, this is true for $n\geq 7$ (in other words for $n\geq 7$, $\sqrt{\pi}\sqrt{\lfloor\frac{3n-1}{2}\rfloor}> 2^4$ and therefore, $\binom{3n-1}{n-2}<\frac{2^{3n-1}}{\sqrt{\pi}\sqrt{\lfloor\frac{3n-1}{2}\rfloor}}<2^{3n-5}$).

\textbf{Step 4: For $L\geq 2$, bound 1 is better than the conjecture (and therefore, also better than $B_{\gamma}$})\\
From the preceding step, we know that for $L=2$ and arbitrary $n$,  $\binom{(L+1)n-1}{n-2}<2^{(L+1)(n-1)-2}$. It remains to show that the right hand side of this inequality increases more rapidly than the left hand side when $L$ increases (and $n$ is fixed). When increasing $L$ to $L+1$ the left hand side increases by the factor 
\begin{align}\label{L+1}
\frac{((L+2)n-1)!\cdot (Ln+1)!}{((L+1)n-1)!\cdot ((L+1)n+1)!}=\frac{(L+1)n\cdots ((L+2)n-1)}{(Ln+2)\cdots ((L+1)n+1)}<(1,5)^n.
\end{align} 
This inequality is true since $2(Ln+2)=(L+1)n+(L-1)n+4$ and hence $\frac{(L+1)n}{Ln+2}=2-\frac{(L-1)n+4}{Ln+2}<2-\frac{L-1}{L}=1+\frac{1}{L}\leq 1,5$. This implies $\frac{(L+1)n+i}{Ln+2+i}<1,5$ for $i\in\{0,\hdots,n-1\}$ and \eqref{L+1} follows.\\
However, when increasing from $L$ to $L+1$, the right hand side increases by the factor $2^{n-1}>(1,5)^n$ for $n\geq 3$.

\textbf{Step 5: The case $L=1$}\\
\textbf{5.1: The case $n-k\nmid\delta$}\\
If $n-k\nmid\delta$, the bound of the conjecture is $2^{2n-2+k+r-1}\geq 2^{2n-2}$. Furthermore, $\binom{2n-1}{n-2}<\frac{2^{2n-1}}{\sqrt{\pi}\sqrt{\lfloor\frac{2n-1}{2}\rfloor}}<2^{2n-2}$ for $n\geq 3$.\\
\textbf{5.2: The case $n-k\mid\delta$ and $k>\delta+1$}\\
$L=1$ and $n-k\mid\delta$ imply $n-k=\delta$. If $k>\delta+1$, one has $2(n-k)-1\leq n-3$ and hence $\binom{2n-1}{2(n-k)-1}\leq \binom{2n-1}{n-3}$. When $n$ increases to $n+1$, the right hand side of this inequality increases by the factor $\frac{(2n+1)!\cdot (n-3)!\cdot (n+2)!}{(n-2)!\cdot(n+3)!\cdot(2n-1)!}=\frac{(2n+1)2n}{(n-2)(n+3)}=\frac{4n^2+2n}{n^2+n-6}$. The conjectured bound, which is in this case equal to $2^{2n-4}$, increases by the factor $4$. It holds $\frac{4n^2+2n}{n^2+n-6}\leq 4\Leftrightarrow n\geq 12$. Moreover, one could compute directly that $\binom{2n-1}{n-3}<2^{2n-4}$ for $n\in\{3,\hdots,12\}$. Consequently,  bound 1 is smaller than the conjectured bound for all $n\geq 3$.\\
\textbf{5.3: The case $n-k\mid\delta$ and $k=\delta+1$}\\
Since here $n=\delta+k=2\delta+1$, one has only to consider odd values for $n$. For $n=16$, one has $\binom{2n-1}{n-2}<2^{2n-4}$ and the left hand side is increasing by the factor $\frac{(2n+1)!\cdot (n-2)!\cdot (n+1)!}{(n-1)!\cdot(n+2)!\cdot(2n-1)!}=\frac{(2n+1)2n}{(n-1)(n+2)}=\frac{4n^2+2n}{n^2+n-2}\leq 4\Leftrightarrow n\geq 4$. Thus, it only remains to consider $(n,(n+1)/2,(n-1)/2)$ convolutional codes for $n\in\{3,5,7,9,11,13,15\}$. Since for $n=3$, one has $\gamma=4<5$, the conjecture does not hold for this parameter. For the other values, one could compute directly that bound 1 is in all cases larger than the conjectured bound but bound 2 is always smaller than the conjectured bound.

\textbf{Step 6: Comparison with $B_{\gamma}$ for the cases in which bound 1 is larger than the conjectured bound}\\
We have to consider $(n,(n+1)/2,(n-1)/2)$ convolutional codes for $n\in\{3,5,7,9,11,13,15\}$. For $n=3$, $B_{\gamma}=4$, bound 1 is equal to $5$ and bound 2 is equal to $2$. For $n=5$, $B_{\gamma}>\frac{1}{2(2n-2)}\binom{4n-6}{2n-3}$ is larger than bound 1 and since it is growing more rapidly in $n$ than bound 1, it is larger than bound 1 for $n\geq 5$ (it is growing with factor $\frac{(2n-2)(4n-3)(4n-4)(4n-5)(4n-6)}{2n(2n-1)^2(2n-2)^2}=\frac{(4n-3)(4n-5)(4n-6)}{n(2n-1)^2}>\frac{(4n-3)(4n-5)}{n(2n-1)}=\frac{(4n-3)(4n-5)}{n(2n-1)}>\frac{16n^2-32n+15}{2(n^2+n-2)}\geq \frac{8n^2+8n+15}{2(n^2+n-2)}\geq\frac{2(4n^2+2n)}{2(n^2+n-2)}$, which is the growing factor of bound 1).

\textbf{Step 7: Comparison with exact minimal values for superregularity}\\
Relevant are codes whose parameters fullfil $(L+1)(n-1)+k+r-1\leq 10$ and $L\geq 1$ (since the case $L=0$ was already considered before). Hence one has to investigate the cases $(3,2,1)$, $(4,3,2)$, $(5,3,2)$, $(5,4,1)$, $(6,5,1)$ and $(6,4,2)$. For $(5,3,2)$, the exact minimal value for superregularity yields $|\mathbb F|\geq 31$, bound 1 yields $|\mathbb F|\geq 37$, the conjectured bound yields $|\mathbb F|\geq 67$, bound 2 yields $|\mathbb F|\geq 87$ and $B_{\gamma}$ yields $|\mathbb F|\geq 233$. For $(3,2,1)$, the exact minimal value for superregularity yields $|\mathbb F|\geq 5$, which is the bound as using $B_{\gamma}$. As seen in the preceding step, bound 1 yields here $|\mathbb F|\geq 7$ and bound 2 gives $|\mathbb F|\geq 3$, which is optimal (the conjecture cannot be applied here since $\gamma=4<5$). For $(4,3,2)$, $(5,4,1)$ and $(6,5,1)$ bound 1 is better than using the minimal value for superregularity. Finally, for $(6,4,2)$, the minimal value for superregularity lies between bound 2 and bound 1.
\end{proof}

\begin{remark}\ \\
For $(n,n-1,\delta)$ with $\delta\leq n-2$, even the bound
$(e-1)\cdot n^{\delta}\cdot\delta !$ is smaller than the conjectured bound $2^{(\delta+1)(n-1)-2}$.
\end{remark}

\begin{proof}\ \\
For a given degree $\delta$, the smallest possible value for $n$ fulfilling the restrictions is $n=\delta+2$. In this case and as long as $\delta\geq 2$ (for $\delta=1$, the problem is solved anyway), $(e-1)\cdot n^{\delta}\cdot\delta !=(e-1)\cdot(\delta+2)^{\delta}\cdot\delta !$ is smaller than $2^{(\delta+1)(n-1)-2}=2^{\delta(\delta+2)-1}$.\\
This is true since for $\delta=2$, one has $(e-1)\cdot 32<2^7$ and the bound increases with factor $(\delta+1)(\delta+2)\left(\frac{\delta+3}{\delta+2}\right)^{\delta}$ when $\delta$ increases by $1$, while the conjecture increases with factor $2^{2\delta+4}$, which is larger because $2^{\delta}>\left(\frac{\delta+3}{\delta+2}\right)^{\delta}$ and $8\cdot 2^{\delta}>(\delta+1)(\delta+2)$.\\
Since the conjectured bound increases with factor $2^{\delta+1}$ when $n$ increases, and our new bound only with factor $(1+1/n)^{\delta}$, the new bound is better for all $2\leq\delta\leq n-2$.
\end{proof}

\subsection{Comparison of bounds for arbitrary parameters}

In this subsection, we only consider cases where $\max(k,n-k)\leq\delta$ since the other cases were already considered before. Using again the duality result from Theorem \ref{dual}, for $B_{\gamma}$ as well as for the bounds of Theorem \ref{M}, one could take the minimum of the values for $(n,k,\delta)$ and $(n,n-k,\delta)$ to get the best possible bound.

\begin{theorem}\ \\
For all parameters with $\delta\geq\max\{k,n-k\}$ but\\
$(2,1,\delta)$ with $\delta$ arbitrary,\\
$(3,k,\delta)$ with $k\in\{1,2\}$, $\delta\in\{2,3,4,5\}$, \\
$(4,2,\delta)$ with $\delta\in\{2,4\}$,\\
$(5,k,3)$ with $k\in\{2,3\}$\\
the bounds of Theorem \ref{M} are able to improve $B_{\gamma}$.
\end{theorem}

\begin{proof}\ \\
Using (amongth others) the Stirling formula, one gets 
$$B_{\gamma}>1/2\cdot \frac{1}{2(L+1)(n-1)}\cdot\frac{2^{2((L+1)(n-1)-1)}}{\sqrt{\pi((L+1)(n-1)-1)}}>\frac{2^{2((L+1)(n-1)-1)}}{4\sqrt{\pi}((L+1)(n-1))^{3/2}}$$
and
\begin{align*}
\binom{(L+1)n}{(L+1)(n-k)}(L+1)(n-k)&<(L+1)(n-k)\cdot\frac{2^{(L+1)n}}{\sqrt{\pi(L+1)(n-1)/2}}\\
&\leq \sqrt{2(L+1)(n-1)}\cdot \frac{2^{(L+1)n}}{\sqrt{\pi}}
\end{align*}
Therefore, in order to get $M_1<B_{\gamma}$, it is sufficient if 
$$\sqrt{2}\cdot 2^{L+5}((L+1)(n-1))^2\leq 2^{(L+1)(n-1)}.$$
It is sufficient to consider the case $L\geq 2$, which is implied by $n-k\leq\delta$ and $k\leq\delta$.\\
For $n=2$, it is clear that above inequaltity is not fulfilled for all $L\geq 2$.\\
For $n=3$, it is fulfilled for $L\geq 14$ (and not for $L\leq 13$) (Mathematica).\\
For $n=4$, it is fulfilled for $L\geq 6$ (and not for $L\leq 5$) (Mathematica).\\
For $n=5$, it is fulfilled for $L\geq 4$ (and not for $L\leq 3$) (Mathematica).\\
For $n=6$, it is fulfilled for $L\geq 3$ (and not for $L\leq 2$) (Mathematica).\\
For $n\geq 7$, it is fulfilled for all $L\geq 2$ (Mathematica).\\
(One can show with Mathematica that it is fulfilled for $L\geq 2$ and $n=7$. When switching from $n$ to $n+1$ the left hand side is growing by the factor $(\frac{n}{n-1})^2\leq 4$, while the right hand side is growing by the factor $2^{L+1}\geq 8$. Therefore, the inequality is fulfilled for $L\geq 2$ and $n\geq 7$.)\\
Since the inequality is only sufficient, it is possible that $M_1$ is better than $B_{\gamma}$ also in other cases than those mentioned above. We check this in the following by computing the bounds directly with Mathematica:\\
For $n=2$, $B_{\gamma}$ is better then $M_1$.\\
For $n=3$ and $L\geq 9$, $M_1$ is better and for $L\leq 7$, $B_{\gamma}$ is better ($L=8$ is not possible with $n=3$).\\
For $n=4$ and $L=5$, $M_1$ is better than $B_{\gamma}$. For $n=4$ and $L=4$, $B_{\gamma}$ is better than $M_1$ if $k=2$ and $\delta=4$ and $M_1$ is better in the other cases. $L=3$ is not possible with $n=4$ and for $n=4$ and $L=2$, $B_{\gamma}$ is better for $\delta=2$ and $M_1$ is better for $\delta=3$.\\
For $n=5$ and $L=3$, $M_1$ is better than $B_{\gamma}$, and for $n=5$ and $L=2$, $B_{\gamma}$ is better then $M_1$.\\
For $n=6$ and $L=2$, $M_1$ is better than $B_{\gamma}$.\\
In all cases for which $B_{\gamma}$ is better than $M_1$ it is also better than $M_2$ and $M_3$ (Mathematica).
Thus, for the following code parameters we are not able to improve $B_{\gamma}$ with the above new bounds (only cases were $\max(k,n-k)\leq\delta$, other case was already considered before):\\
$(2,1,\delta)$ with $\delta$ arbitrary,\\
$(3,k,\delta)$ with $k\in\{1,2\}$, $\delta\in\{2,3,4,5\}$, \\
$(4,2,\delta)$ with $\delta\in\{2,4\}$,\\
$(5,k,3)$ with $k\in\{2,3\}$
\end{proof}

\section{Conclusion}
In this paper, bounds for the probability and the necessary field size for MDP and complete MDP convolutional codes have been shown. Moreover, it has been proven that these bounds on the field size are able to improve the already existing bounds. However, it is clear that these bounds are not optimal and they do not lead to concrete constructions of codes. Hence, this paper could be considered as one step forward towards solving the big problem of determining the exact minimum field size for the existence of MDP and complete MDP convolutional codes and providing constructions of these codes over fields of possibly small size.

\bibliography{mybibfile}

\end{document}